\documentclass[letterpaper, 10 pt, conference]{ieeeconf} 
\IEEEoverridecommandlockouts  
\usepackage{xcolor}
\usepackage{graphicx}      
\usepackage{import}        
\usepackage{amsmath,amsfonts}
\usepackage[hidelinks]{hyperref}
\usepackage{comment}

\usepackage{algorithm}
\usepackage{lmodern} 
\usepackage{algpseudocode}
\usepackage[utf8]{inputenc}
\usepackage{mathtools,halloweenmath}
\usepackage{booktabs} 
\usepackage{amssymb}
\usepackage{graphics} 
\usepackage{float}
\usepackage{epsfig} 
\usepackage{times} 
\usepackage{bm} 
\usepackage{cite}
\usepackage[english]{babel}
\usepackage{mdframed} 

\usepackage{amsthm}       

\let\labelindent\relax 
\usepackage{enumitem}

\newcommand{\mR}{{\mathbb R}}

\theoremstyle{plain}
\newtheorem{theorem}{Theorem}

\newtheorem{proposition}{Proposition}
\newtheorem{corollary}{Corollary}
\theoremstyle{definition}
\newtheorem{assumption}{Assumption}
\newtheorem{remark}{\textbf{Remark}}
\newtheorem{property}{Property}

\title{\LARGE Spectral Koopman-Hopf Formula  for Reachability with Adversary}
\author{Sarang Sutavani and Umesh Vaidya
\thanks{UV will like to acknowledge the financial support from CMMI Award 2531804. Sarang Sutavani is with the Dept. of Electrical Engr., Clemson University, Clemson SC. Email: ssutava@clemson.edu. Umesh Vaidya is with the Dept. of Mechanical Engr., Clemson University, Clemson SC. Email: uvaidya@clemson.edu. 
}
}
\begin{document}
\maketitle
\thispagestyle{empty}
\pagestyle{empty}

\begin{abstract}
This paper develops a spectral Koopman-Hopf framework for adversarial reachability analysis of nonlinear systems. By lifting the nonlinear drift dynamics into Koopman eigenfunction coordinates, the proposed approach transforms the original state-dependent Hamilton-Jacobi-Isaacs (HJI) equation into an approximate state-independent optimization problem in spectral coordinates. The transformed control and disturbance directions are approximated using least-squares spectral projections, and upper and lower bounds on the transformed Hamiltonian are used to derive corresponding bounds on the value function and backward reachable sets. Koopman eigenfunctions are computed using a path-integral-based Galerkin framework that avoids spatial discretization of the associated eigenfunction PDEs. Numerical examples demonstrate tractable approximation of backward reachable sets for nonlinear adversarial systems.
\end{abstract}


\section{Introduction}

Adversarial reachability for nonlinear systems is naturally formulated as a two-player differential game, where the backward reachable set is characterized by the viscosity solution of a Hamilton-Jacobi-Isaacs (HJI) equation \cite{bardi1997optimal}. Despite its theoretical completeness, numerical solution of HJI equations remains fundamentally limited by the curse of dimensionality, as grid-based methods scale exponentially with the state dimension \cite{mitchell2005time, mitchell2007comparing}. This limitation severely restricts the applicability of classical dynamic programming approaches to low-dimensional systems.

{\color{black}
The Hopf formula provides an alternative to full state-space discretization of the HJI PDE by representing the Hamilton-Jacobi solution as a pointwise optimization problem \cite{hopf1965generalized, bardi1984hopf}, enabling efficient parallel evaluation \cite{darbon2016algorithms, chow2019algorithm}.}
Recent works extend this paradigm to broader classes of problems, including characteristic-based formulations and generalized Hopf representations for state-dependent Hamiltonians \cite{yegorov2021perspectives, lee2020hopf, chen2024hopf}. However, these approaches typically involve solving auxiliary optimization problems along characteristics or introducing additional complexity in the dual formulation. For general nonlinear control-disturbance systems, the Hamiltonian remains state-dependent, and the adversarial min-max structure further limits the use of classical Hopf methods. 
Recent literature has explored the application of the generalized Hopf formula to nonlinear systems, though often without rigorous theoretical guarantees. For instance, Kirchner et al. \cite{kirchner2017time} applied the formula to nonlinear pursuit-evasion models using direct Jacobian linearization. More recently, Sharpless et al. \cite{sharpless2023koopman} employed Koopman-based lifting to map the dynamics and safe sets into higher-dimensional spaces for improved approximation. 
{\color{black}In contrast, our approach utilizes a nonlinear spectral transformation providing a more informed formulation than standard Jacobian-based linearization. If $n$ principal Koopman eigenfunctions are involved in the transformation then the dimensionality of the original problem is maintained.}

Koopman operator theory offers a complementary perspective by representing nonlinear dynamics through linear evolution in a space of observables \cite{mezic2020spectrum}. In particular, Koopman eigenfunctions define intrinsic coordinates in which the drift dynamics evolve linearly, providing a natural mechanism for uncovering latent linear structure in nonlinear systems \cite{korda2018linear, korda2020optimal, vaidya2025koopman}. In our recent work \cite{sutavani2025spectral}, we exploited this property to construct a spectral Koopman-Hopf framework for nonlinear optimal control without adversarial inputs, where a coordinate transformation renders the Hamiltonian state-independent under a rigidity condition, enabling direct application of the Hopf formula without state-space discretization.

This paper extends that framework to adversarial reachability governed by HJI equations. By lifting the nonlinear system into Koopman eigenfunction coordinates, the drift becomes linear and decoupled, and under a spectral and approximate spectral input rigidity condition, both control and disturbance directions are mapped to constant matrices. This transformation converts the nonlinear differential game into a representation with a state-independent Isaacs Hamiltonian, enabling a finite-dimensional Hopf-type approximation of the value function for nonlinear reachability problems.

A key component used is a path-integral construction of Koopman eigenfunctions, which replaces the need to solve eigenfunction PDEs on a spatial grid with trajectory-based evaluations \cite{deka2023path}. By converting the eigenfunction computation into integrals along  trajectories, the approach avoids discretization of the nonlinear HJI equation and the associated linear eigenfunction PDE. 
Consequently, the curse of dimensionality is shifted from the state space to a spectral approximation problem involving trajectory simulations.

The resulting spectral Koopman-Hopf framework provides a scalable, pointwise optimization based approach for adversarial reachability that combines operator-theoretic structure with convex Hamilton-Jacobi representations while preserving the original system dimension when only the $n$ principal eigenfunctions are involved. Using $N>n$ modes improves rigidity at higher optimization cost. 
Numerical examples demonstrate that the method enables tractable computation of reachable sets for nonlinear systems with nontrivial drift dynamics and adversarial inputs.

\section{Preliminaries: Koopman Spectral Theory}

Consider the autonomous nonlinear system
\begin{equation}
    \dot x = f(x), \qquad x\in\mathcal{X}\subset\mathbb{R}^n.
\end{equation}
Let $s_t(x)$ denote the flow map of $f$. 
The Koopman operator $\mathcal{K}^t$ acts linearly on observables 
$\psi:\mathcal{X}\to\mathbb{C}$ via
\[
(\mathcal{K}^t \psi)(x) := \psi(s_t(x)).
\]
Its infinitesimal generator is the Lie derivative
$\mathcal{L}_f \psi(x) = \nabla \psi(x)^\top f(x)$.
A function $\phi$ is a Koopman eigenfunction with eigenvalue $\lambda\in\mathbb{C}$ if
\begin{equation}
    \nabla \phi(x)^\top f(x) = \lambda \phi(x).
\end{equation}
Along trajectories, $
\phi(s_t(x)) = e^{\lambda t}\phi(x)$.
Thus Koopman eigenfunctions evolve linearly in time even when the state dynamics are nonlinear.
Suppose $x=0$ is a hyperbolic equilibrium of $f$, with linearization $A = Df(0)$. 
Under standard smooth conjugacy results for hyperbolic equilibria, there exist $n$ principal eigenfunctions $\phi_i$ whose eigenvalues coincide with the eigenvalues of $A$, and whose gradients at the origin align with the left eigenvectors of $A$. Defining
\[
\Phi(x) := (\phi_1(x),\dots,\phi_n(x))^\top,
\]
the map $\Phi$ defines a $C^1$ local diffeomorphism on a neighborhood $\mathcal{P}$ of the equilibrium.
In these coordinates 
\[
\frac{d}{dt}\Phi(x(t)) = \Lambda \Phi(x(t)),
\qquad
\Lambda = \mathrm{diag}(\lambda_1,\dots,\lambda_n).
\]
This spectral coordinate transformation will be the foundation for the Hopf reachability formulation.
In the reachability construction, we use real-valued spectral coordinates. For complex-conjugate Koopman eigenpairs, the real and imaginary parts of the corresponding eigenfunctions are used to form real coordinates.
The Koopman principal eigenfunctions enjoy the following property.

\begin{property}\label{property1}
    Let $\phi_{\lambda_1}$ and $\phi_{\lambda_2}$ be the eigenfunctions associated with eigenvalues $\lambda_1$ and $\lambda_2$ respectively. If $\phi_{\lambda_1}^{k_1}\phi_{\lambda_2}^{k_2}\in {\cal C}^1$, for $k_1,k_2\in \mR^{+}$, then it is an eigenfunctions of the Koopman generator with eigenvalue $k_1\lambda_1+k_2\lambda_2$. 
\end{property}

\section{Adversarial Reachability Formulation}

Consider the nonlinear control-disturbance system
\begin{equation}
    \dot x = f(x) + G(x)u + E(x)d,
    \qquad x \in \mathcal{X}\subset\mathbb{R}^n,
    \label{eq:system}
\end{equation}
where $u(\cdot)$ is the control input and $d(\cdot)$ is an adversarial disturbance.
The admissible sets $\mathcal{U}\subset\mathbb{R}^{m}$ and $\mathcal{D}\subset\mathbb{R}^{p}$ are compact, convex, and contain the origin. Let $\mathcal{T}\subset\mathcal{X}$ denote the target set.
We choose a continuous terminal cost $J:\mathcal{X}\to\mathbb{R}$ such that
\begin{align}
    &J(x)<0 \ \text{for } x\in\mathrm{int}(\mathcal{T}),\;\;
    J(x)=0 \ \text{for } x\in\partial\mathcal{T},\nonumber\\
    &J(x)>0 \ \text{for } x\in\mathcal{X}\setminus\mathcal{T}.
    \label{eq:terminal_cost_target}
\end{align}

To enforce the hard input constraints $u(t)\in\mathcal{U}$ and $d(t)\in\mathcal{D}$, we define the extended-value indicator function
\begin{equation}
    \mathbb{I}_{\mathcal{C}}(z)
    =
    \begin{cases}
        0, & z\in\mathcal{C},\\
        +\infty, & z\notin\mathcal{C},
    \end{cases}
    \label{eq:indicator}
\end{equation}
and consider the finite-horizon differential game with running cost $\Gamma(u(s),d(s)) = \mathbb{I}_{\mathcal{U}}(u(s))-\mathbb{I}_{\mathcal{D}}(d(s))$ 
\begin{equation}
    \begin{aligned}
        V(x,t)
        =
        \sup_{d(\cdot)}\;\inf_{u(\cdot)}
        \Bigg\{
        \int_t^T \Gamma(u(\tau),d(\tau)) d\tau + J(x(T)) \Bigg\}
    \end{aligned}
    \label{eq:reach_value_indicator}
\end{equation}
subject to \eqref{eq:system} with initial condition $x(t)=x$. Because the integrand in \eqref{eq:reach_value_indicator} is finite if and only if $u(\tau)\in\mathcal{U}$ and $d(\tau)\in\mathcal{D}$ for almost all $\tau\in[t,T]$, the value function \eqref{eq:reach_value_indicator} is equivalent to the constrained terminal-cost game
\begin{equation}
    V(x,t)
    =
    \sup_{d(\cdot)\in\mathcal{D}}\;\inf_{u(\cdot)\in\mathcal{U}} J(x(T)),
    \label{eq:reach_value_constrained}
\end{equation}
which is the standard reachability formulation. The controller seeks to drive the state into $\mathcal{T}$ by time $T$ despite the worst-case disturbance. The backward reachable set is characterized by the zero sublevel set of the value function:
\begin{equation}
    \mathcal{R}(t)
    =
    \{ x \in \mathcal{X} \mid V(x,t) \le 0 \}.
    \label{eq:brs_def}
\end{equation}
Thus, $V(x,t)\le 0$ if and only if the controller can force the system into $\mathcal{T}$ at time $T$ while satisfying the hard constraints $u(\cdot)\in\mathcal{U}$ and $d(\cdot)\in\mathcal{D}$ under the worst-case disturbance.


\subsection{Hamiltonian and HJI PDE}

The indicator running cost in \eqref{eq:reach_value_indicator} enforces the hard constraints $u(\cdot)\in\mathcal{U}$ and $d(\cdot)\in\mathcal{D}$ almost everywhere. 
The associated differential game is assumed to satisfy.

\begin{assumption}[Isaacs condition]\label{assume:isaacs_condition}
    The min-max and max-min Hamiltonians coincide for all $(x,p)\in\mathbb R^n\times\mathbb R^n$, i.e.,
    \begin{equation}\label{eq:isaacs_condition}
        \begin{aligned}
            \sup_{d\in\mathcal D}\;\inf_{u\in\mathcal U}
            &\ p^\top\!\big(f(x)+G(x)u+E(x)d\big)
            \\
            &=
            \inf_{u\in\mathcal U}\;\sup_{d\in\mathcal D}
            \ p^\top\!\big(f(x)+G(x)u+E(x)d\big).
        \end{aligned}
    \end{equation}
\end{assumption}
Under Assumption \ref{assume:isaacs_condition}, the dynamic programming principle yields the Hamilton-Jacobi-Isaacs (HJI) equation for the value function.
Writing $p=\nabla_x V(x,t)$, the Hamiltonian is
\begin{equation}
    H(x,p)
    =
    \inf_{u\in\mathcal{U}}\;\sup_{d\in\mathcal{D}}
    p^\top\big(f(x)+G(x)u+E(x)d\big).
    \label{eq:hamiltonian}
\end{equation}
Consequently, the value function, $V$, satisfies the HJI PDE
\begin{equation}
    \frac{\partial V}{\partial t}(x,t) + H\!\left(x,\nabla_x V(x,t)\right)=0,
    \quad V(x,T)=J(x),
    \label{eq:hji}
\end{equation}
in the viscosity sense. The backward reachable set at time $t$ is then given by the zero sublevel set $\mathcal{R}(t)=\{x\mid V(x,t)\le 0\}$, and the controller can guarantee reachability of $\mathcal{T}$ at time $T$ against all disturbances $d(\cdot)\in\mathcal{D}$ if and only if $V(x,t)\le 0$. Whenever the minimizer and maximizer in \eqref{eq:hamiltonian} exist, the corresponding optimal feedback policies can be expressed as
\begin{equation}
    \begin{aligned}
        u^\star(x,t)\in\arg\min_{u\in\mathcal{U}}
        \nabla_x V(x,t)^\top G(x)u,
        \\
        d^\star(x,t)\in\arg\max_{d\in\mathcal{D}}
        \nabla_x V(x,t)^\top E(x)d,
    \end{aligned}\label{eq:opt_controls}
\end{equation}
where the drift term $\nabla_x V^\top f(x)$ does not affect the minimization/maximization since it is independent of $u$ and $d$.
For common compact convex sets (e.g., boxes or Euclidean balls), the optimizers in \eqref{eq:opt_controls} admit closed-form expressions as projections/saturations determined by the sign of the corresponding components of $G(x)^\top\nabla_x V$ and $E(x)^\top\nabla_x V$.

\section{Spectral Koopman Hopf Formula}

We utilize Koopman spectral coordinates to eliminate state dependence in the Hamiltonian, yielding a tractable representation of the Hamilton–Jacobi–Isaacs (HJI) equation:
\[
\frac{\partial V}{\partial t}
+
\inf_{u\in\mathcal U}
\sup_{d\in\mathcal D}
\nabla_x V^\top\!\left(f(x)+G(x)u+E(x)d\right)
=
0.
\]
Let $\Phi(x)$ denote the principal Koopman eigenfunction map satisfying
$\frac{\partial \Phi}{\partial x} f(x) = \Lambda \Phi(x)$.
We introduce the time-varying spectral coordinate $X = e^{-\Lambda t}\Phi(x)$ that renders the drift dynamics stationary in spectral space.

\subsection{Transformed Input Directions in Spectral Coordinates}

Let
$
\Phi_N(x)=\begin{bmatrix}
\phi_1(x) & \cdots & \phi_N(x)
\end{bmatrix}^\top
$
denote a truncated Koopman eigenfunction map consisting of the principal eigenfunctions together with higher-order eigenfunctions generated using the algebra property of Koopman eigenfunctions (Property \ref{property1}). 
{\color{black}The retained eigenfunctions are ordered according to the magnitude of the real part of their associated Koopman eigenvalues, with the principal eigenfunctions corresponding to the dominant spectral modes near the equilibrium.}
The spectral coordinates are
\[
X=e^{-\Lambda_N t}\Phi_N(x),
\;\;\Lambda_N=\mathrm{diag}(\lambda_1,\dots,\lambda_N).
\]
Using the chain rule,
{\color{black}
\begin{align*}
    \nabla_x V(x,t)
    =
    \left(\frac{\partial \Phi_N}{\partial x}(x)\right)^\top
    e^{-\Lambda_N^\top t}
    \nabla_X \bar V(X,t).
\end{align*}
}
Define the transformed control and disturbance directions
\begin{equation}
    M_{u,N}(x):=
    \frac{\partial\Phi_N}{\partial x}(x)G(x),
    \
    M_{d,N}(x):=
    \frac{\partial\Phi_N}{\partial x}(x)E(x).
    \label{eq:MuMd_def}
\end{equation}
Substituting into the Hamiltonian yields
\begin{equation}
    \bar H(X,P,t)
    =
    \inf_{u\in\mathcal U}
    \sup_{d\in\mathcal D}
    P^\top e^{-\Lambda_N t}
    \left(
    M_{u,N}(x)u + M_{d,N}(x)d
    \right),
    \label{eq:Hbar_general}
\end{equation}
where $P=\nabla_X \bar V(X,t)$. 
The transformed Hamiltonian remains state dependent through the matrices $M_u(x)$ and $M_d(x)$.
The objective of the proposed framework is to approximate or bound these transformed input directions using state-independent representations.

\subsection{Approximate Spectral Rigidity}\label{sec:approx_rigidity}

We seek constant matrices
\(
B_{u,N}\in\mathbb R^{N\times m},
\
B_{d,N}\in\mathbb R^{N\times p},
\)
that best approximate the transformed input directions
over a neighborhood $\mathcal P$.
Specifically, we solve the least-squares projection problem
\begin{equation}
    \begin{aligned}
        B_{u,N}
        &=
        \arg\min_{B\in\mathbb R^{N\times m}}
        \int_{\mathcal P}
        \|B-M_{u,N}(x)\|_F^2 \rho(x)\,dx,
        \\
        B_{d,N}
        &=
        \arg\min_{B\in\mathbb R^{N\times p}}
        \int_{\mathcal P}
        \|B-M_{d,N}(x)\|_F^2 \rho(x)\,dx,
    \end{aligned} \label{eq:BuBd_fit}
\end{equation}
where $\rho(x)$ is a sampling density on $\mathcal P$. Using uniform i.i.d. sampled data $\{x_k\}_{k=1}^K$, the solution reduces to
\begin{equation}
    B_{u,N}=
    \frac1K\sum_{k=1}^K M_{u,N}(x_k),
    \;\;
    B_{d,N}=
    \frac1K\sum_{k=1}^K M_{d,N}(x_k). \label{eq:BuBd_sample_approx}
\end{equation}

The corresponding approximation residuals are 
\begin{equation}
    \begin{aligned}
        \varepsilon_{u,N}
        =
        \sup_{x\in\mathcal P}
        \|M_{u,N}(x)-B_{u,N}\|,
        \\
        \varepsilon_{d,N}
        =
        \sup_{x\in\mathcal P}
        \|M_{d,N}(x)-B_{d,N}\|.
    \end{aligned}
    \label{eq:rigidity_residual}
\end{equation}

Substituting the approximations
$M_{u,N}(x)\approx B_{u,N}$ and $M_{d,N}(x)\approx B_{d,N}$
into \eqref{eq:Hbar_general} yields the approximate
state-independent Hamiltonian \eqref{eq:Hbar_approx} and a pointwise Hopf approximation of the value \eqref{eq:generalized_Hopf_formula}, where $\bar J^*$ is the  Fenchel-Legendre transform of the terminal cost in spectral coordinates.
\begin{equation}
    \bar H_N(P,t)
    =
    \inf_{u\in\mathcal U}
    \sup_{d\in\mathcal D}
    P^\top e^{-\Lambda_N t}
    (B_{u,N} u+B_{d,N} d).
    \label{eq:Hbar_approx}
\end{equation}

\begin{equation}\label{eq:generalized_Hopf_formula}
    \bar V(X,t)
    =
    \sup_P
    \left\{
    P^\top X
    +
    \int_t^T \bar H(P,\tau)\,d\tau
    -
    \bar J^*(P)
    \right\}.
\end{equation}
The original value and backward reachable set follow from $V(x,t)=\bar V(e^{-\Lambda t}\Phi(x),t)$ and $\mathcal R(t)=\{x:V(x,t)\le0\}$. 
Equation \eqref{eq:generalized_Hopf_formula} is a finite-dimensional and evaluated at each query $x$. No HJI state-space grid discretization is required.

For a compact convex set $\mathcal C \subset \mathbb R^m$, its support function is defined by $\sigma_{\mathcal C}(y) = \sup_{c\in\mathcal C} y^\top c$.
Using this notation, the approximate Hamiltonian
\eqref{eq:Hbar_approx} can be written as
\begin{equation*}
    \bar H_N(P,t)
    =
    -
    \sigma_{\mathcal U}
    \!\left(
    -B_{u,N}^\top e^{-\Lambda_N t}P
    \right)
    +
    \sigma_{\mathcal D}
    \!\left(
    B_{d,N}^\top e^{-\Lambda_N t}P
    \right).
\end{equation*}

\begin{remark}[Structure of the transformed Hamiltonian]
    The approximated Hamiltonian \eqref{eq:Hbar_approx} takes the form
    \(\bar H_N(P,t) = - \sigma_{\mathcal U}(\cdot) + \sigma_{\mathcal D}(\cdot)\), and is typically a difference-of-convex functions in the costate $P$.
    Hence, the resulting Hopf optimization problem need not be globally convex.
\end{remark}

Because \eqref{eq:generalized_Hopf_formula} can be nonconvex, it is treated as a numerical approximation unless global optimality is certified.
Using $N>n$ modes can reduce the residuals $\epsilon_{u,N}$ and $\epsilon_{d,N}$, but increases the costate dimension to $N$. 
However, unlike $O(M^n)$ grid-based HJI, each query still requires only an $N$-dimensional optimization.

\begin{assumption}[Spectral approximation consistency]
    Let $\Phi_N$ be the truncated Koopman eigenfunction dictionary used to construct the transformed control and disturbance inputs maps $M_{u,N}$ and $M_{d,N}$ in \eqref{eq:MuMd_def}. 
    There exist limiting transformed input maps $M_u,M_d$ such that \( \|M_{u,N}-M_u\|_{L^\infty(\mathcal P)}\to 0, \ \|M_{d,N}-M_d\|_{L^\infty(\mathcal P)}\to 0 \)
    as $N\to\infty$.
\end{assumption}

\begin{proposition}[Hamiltonian approximation error]\label{thm:propos1}
    Assume $\mathcal U$ and $\mathcal D$ are compact and let
    \(
    U_{\max}:=\sup_{u\in\mathcal U}\|u\|,
    \
    D_{\max}:=\sup_{d\in\mathcal D}\|d\|.
    \)
    Then the approximate Hamiltonian $\bar H_N$ \eqref{eq:Hbar_approx} and the transformed Hamiltonian \eqref{eq:Hbar_general} satisfy
    \begin{equation}
        \begin{aligned}
            |\bar H(X,P,t)-&\bar H_N(P,t)| \le \\
            &\|e^{-\Lambda_N^\top t}P\|
            \left(
            U_{\max}\varepsilon_{u,N}
            +
            D_{\max}\varepsilon_{d,N}
            \right).
        \end{aligned}
    \end{equation}
\end{proposition}
\begin{proof}
    By the minimax property $|\inf_u \sup_d h_1 - \inf_u \sup_d h_2| \le \sup_{u,d} |h_1 - h_2|$ and triangle inequality, 
    \begin{align*}
        |\bar H - \bar H_N| 
        &\le \sup_{u\in\mathcal U} \left| P^\top e^{-\Lambda_N t} (M_{u,N}(x) - B_{u,N})u \right| \\
        &+ \sup_{d\in\mathcal D} \left| P^\top e^{-\Lambda_N t} (M_{d,N}(x) - B_{d,N})d \right|.
    \end{align*}
    Applying Cauchy-Schwarz gives the stated bound.
\end{proof}

\begin{remark}[Convergence and spectral approximation]
    Proposition~1 shows that the Hamiltonian approximation error is controlled by $\|e^{-\Lambda_N^\top t}P\| (U_{\max}\epsilon_{u,N}+D_{\max}\epsilon_{d,N})$.
    Thus, if this quantity vanishes as $N\to\infty$ uniformly on compact costate-time sets, then $\bar H_N$ converges locally uniformly to $\bar H$. Under standard stability and comparison assumptions for viscosity solutions, together with consistent terminal data, the corresponding value functions also converge. 
    A general guarantee of vanishing residuals requires additional approximation properties of the Koopman eigenfunction dictionary, which are not established here. Hence, $\epsilon_{u,N}$ and $\epsilon_{d,N}$ serve as i
    .ndicators of spectral approximation accuracy.
\end{remark}

\subsection{State-Independent Bounds on the Hamiltonian} \label{sec:reachability_bounds}

We now derive state-independent upper and lower bounds
for the transformed Hamiltonian. Assume ellipsoidal sets
\[
\mathcal U
=
\{u:u^\top R^{-1}u\le1\},
\;\;
\mathcal D
=
\{d:d^\top S^{-1}d\le1\}.
\]
\begin{equation}
    \begin{aligned}
            Q_u(x):= M_{u,N}(x)R M_{u,N}(x)^\top,
            \\
            Q_d(x):= M_{d,N}(x)S M_{d,N}(x)^\top.
    \end{aligned}
    \label{eq:QuQd}
\end{equation}

Assume there exist constant matrices satisfying
\begin{equation}
    \underline Q_u
    \preceq
    Q_u(x)
    \preceq
    \overline Q_u,
    \
    \underline Q_d
    \preceq
    Q_d(x)
    \preceq
    \overline Q_d,
    \
    \forall x\in\mathcal P.
    \label{eq:Q_bounds}
\end{equation}

\begin{theorem}
Under \eqref{eq:Q_bounds}, the transformed Hamiltonian
\eqref{eq:Hbar_general} admits state-independent bounds
\[
\bar H^\ell(P,t)
\le
\bar H(X,P,t)
\le
\bar H^u(P,t),
\]
where
\begin{equation}
    \begin{aligned}
        \bar H^u(P,t) = - \left\| e^{-\Lambda_N t}P \right\|_{\underline Q_u} + \left\| e^{-\Lambda_N t}P \right\|_{\overline Q_d},\\
        \bar H^\ell(P,t) = - \left\| e^{-\Lambda_N t}P \right\|_{\overline Q_u} + \left\| e^{-\Lambda_N t}P \right\|_{\underline Q_d}.
    \end{aligned}
\end{equation}

\end{theorem}

\begin{proof}
    The support functions of ellipsoids satisfy
    $\sigma_{\mathcal U}(y) = \|y\|_{R}$ and $\sigma_{\mathcal D}(y) = \|y\|_{S}$.
    Substituting into \eqref{eq:Hbar_general} gives $\bar H = - \left\| e^{-\Lambda_N t}P \right\|_{Q_u(x)} + \left\| e^{-\Lambda_N t}P \right\|_{Q_d(x)}$.
    The result follows from monotonicity of the quadratic form
    with respect to the Loewner ordering.
\end{proof}

\begin{theorem}[Value function bounds]
    Let $\bar V$ denote the viscosity solution of
    \[
    \partial_t \bar V(X,t)+\bar H(X,\nabla_X\bar V,t)=0,
    \qquad
    \bar V(X,T)=\bar J(X).
    \]
    Suppose there exist state-independent Hamiltonians
    $\bar H^\ell(P,t)$ and $\bar H^u(P,t)$ such that
    \[
    \bar H^\ell(P,t)
    \le
    \bar H(X,P,t)
    \le
    \bar H^u(P,t),
    \qquad
    \forall (X,P,t).
    \]
    Let $\bar V^\ell$ and $\bar V^u$ be the viscosity solutions
    \begin{equation*}
        \begin{aligned}
            \partial_t \bar V^{\ell} + \bar H^{\ell}(\nabla_X\bar V^{\ell},t)=0,&
            \;
            \bar V^{\ell}(X,T)=\bar J(X),
            \\
            \partial_t \bar V^u + \bar H^u(\nabla_X\bar V^u,t)=0,&
            \;
            \bar V^u(X,T)=\bar J(X).
        \end{aligned}
    \end{equation*}
    Then, under the comparison principle for viscosity solutions,
    \[
    \bar V^\ell(X,t)
    \le
    \bar V(X,t)
    \le
    \bar V^u(X,t),
    \qquad
    (X,t)\in\mathcal P\times[0,T].
    \]
    Consequently, in the original coordinates,
    \[
    V^\ell(x,t)
    \le
    V(x,t)
    \le
    V^u(x,t),
    \]
    where
    \[
    V^\ell(x,t)=\bar V^\ell(e^{-\Lambda t}\Phi(x),t),
    \;\;
    V^u(x,t)=\bar V^u(e^{-\Lambda t}\Phi(x),t).
    \]
\end{theorem}

\begin{proof}
    The result follows from the viscosity comparison principle.
    Since $\bar H^\ell\le \bar H\le \bar H^u$ and all three value functions share the same terminal condition $\bar J$, the solution associated with the smaller Hamiltonian provides a lower solution, while the solution associated with the larger Hamiltonian provides an upper solution. Hence $\bar V^\ell\le \bar V\le \bar V^u$.
    The result in the original coordinates follows by the spectral change of variables $X=e^{-\Lambda t}\Phi(x)$.
\end{proof}
\begin{corollary}[Reachable set enclosure]
    Let $\mathcal R(t)=\{x:V(x,t)\le0\}$. Define
    \[
    \mathcal R^u(t)=\{x:V^u(x,t)\le0\},
    \;\;
    \mathcal R^\ell(t)=\{x:V^\ell(x,t)\le0\}.
    \]
    Then $\mathcal R^u(t) \subseteq \mathcal R(t) \subseteq \mathcal R^\ell(t)$.
\end{corollary}

\section{Koopman Eigenfunction Approximation} \label{sec:eigenfunction_compute}


To compute the principal Koopman eigenfunctions in a curse of dimensionality free manner, we adopt the linear-nonlinear decomposition and finite-time path-integral construction used in \cite{deka2023path}. Assume the drift admits a decomposition on a neighborhood \(\mathcal P\) of the equilibrium,
\begin{equation}\label{eq:split}
f(x)=Ax+f_n(x), \qquad f_n(0)=0,\quad Df_n(0)=0,
\end{equation}
where \(A:=Df(0)\) is the Jacobian at the equilibrium and \(f_n\) denotes the nonlinear remainder.
Let \((\lambda_i,w_i)\) be a (right) eigenpair of \(A^\top\), i.e., $A^\top w_i = \lambda_i w_i$.
For the linear system \(\dot x=Ax\), the corresponding Koopman eigenfunction is
\(\phi_i^{\mathrm{lin}}(x)=w_i^\top x\).
For the nonlinear system \(\dot x=f(x)\), we seek a principal eigenfunction of the form
\begin{equation}\label{eq:phi_decomp}
    \phi_i(x)=\phi_i^{\mathrm{lin}}(x)+\eta_i(x)
    = w_i^\top x+\eta_i(x),
\end{equation}
where \(\eta_i\) is a nonlinear correction satisfying \(\eta_i(0)=0\) and \(\nabla \eta_i(0)=0\).
Substituting \eqref{eq:phi_decomp} into the Koopman eigenfunction PDE
\(\nabla\phi_i(x)^\top f(x)=\lambda_i \phi_i(x)\) yields a first-order linear PDE for \(\eta_i\):
\begin{equation}\label{eq:eta_pde}
    \nabla \eta_i(x)^\top f(x) - \lambda_i \eta_i(x)
    =
    - w_i^\top f_n(x).
\end{equation}
Thus, \(\eta_i\) satisfies a forced transport equation along trajectories of the nonlinear drift. Let \(x(\tau)=s_\tau(x)\) denote the drift trajectory starting at \(x\) at time \(0\).
Multiplying \eqref{eq:eta_pde} by \(e^{-\lambda_i \tau}\) and integrating along characteristics
yields the finite-time path-integral identity
\begin{equation}\label{eq:path_integral_eta_finite}
    \eta_i(x)
    =
    e^{-\lambda_i T}\eta_i(s_T(x))
    +
    \int_{0}^{T} e^{-\lambda_i \tau}\, \big(-w_i^\top f_n(s_\tau(x))\big)\, d\tau,
\end{equation}
for $ T\ge 0$. 
The above representation is not, in general, explicit due to the terminal term
\(e^{-\lambda_i T}\eta_i(s_T(x))\). However, under the assumption that
\(\lim_{T\to \infty} e^{-\lambda_i T}\eta_i(s_T(x))=0\), one obtains the explicit formula
\begin{equation}\label{eq:path_integral_eta_infinite}
    \eta_i(x)
    =
    \int_{0}^{\infty} e^{-\lambda_i \tau}\, \big(-w_i^\top f_n(s_\tau(x))\big)\, d\tau.
\end{equation}
The conditions to ensure \(\lim_{T\to \infty} e^{-\lambda_i T}\eta_i(s_T(x))=0\) are provided in \cite{deka2023path} and can be restrictive.
Motivated by \eqref{eq:path_integral_eta_infinite}, we therefore introduce a family of
path-integral inspired basis functions for approximating \(\eta_i\).
For a given eigenvalue \(\lambda\), 
\begin{equation}\label{eq:path_integral_basis}
    \psi_\ell(x)
    =
    -\int_{0}^{T_\ell} e^{-\lambda \tau}\, w_i^\top f_n(s_\tau(x))\, d\tau,
\end{equation}
for $\ell=1,\ldots,M$, where \(T_\ell>0\) are finite horizons. Each \(\psi_\ell\) is computable from short rollouts
of the drift and does not require gridding the state space.
We then approximate the nonlinear correction using a Galerkin expansion
\begin{align}\label{eq:eta_galerkin}
    \eta_i(x)\approx \sum_{\ell=1}^{M} c_{i,\ell}\, \psi_{\ell}(x)
    =: c_i^\top \Psi(x),
\end{align}
This basis is tailored to the Koopman eigenfunction equation because each \(\psi_\ell\) already encodes
the discounted evolution of the nonlinear part of the vector field along drift trajectories. 
Using snapshot pairs \((x_k,y_k)\) with \(y_k=s_{\Delta t}(x_k)\),
we enforce the finite-time Koopman identity $\phi_i(y_k)\approx e^{\lambda_i \Delta t}\phi_i(x_k)$, with \(\phi_i(x)=w_i^\top x+c_i^\top \Psi(x)\).
This yields a linear least-squares (or projected Galerkin) problem for the coefficients \(c_i\),
which can be solved efficiently from data. Finally, collecting the learned eigenfunctions \(\{\phi_i\}_{i=1}^n\) defines the principal eigenfunction map
\(\Phi(x)=(\phi_1(x),\dots,\phi_n(x))^\top\), which is then used to evaluate the spectral coordinate
\(X=e^{-\Lambda t}\Phi(x)\) pointwise for Hopf reachability.


\section{Simulation Results}
\label{sec:simulations}

\subsection{Two-Dimensional Example: Rigidity and Bounds}
\label{sec:2D_example}

Consider a two-dimensional system governed by $\dot x = f(x) + G(x)u + E(x)d$, where $s_i := \sin(x_i)$, $c_i := \cos(x_i)$, and $d_0 := 2 + c_1 c_2$. The common drift vector field is:
\begin{align}
    f(x) = \frac{1}{d_0} 
    \begin{bmatrix} 
        0.8 c_2 (s_1 - 2x_2) - (x_1 + s_2) \\ 
        -0.8 (s_1 - 2x_2) - 0.5 c_1 (x_1 + s_2) 
    \end{bmatrix}.
\end{align}
The system possesses known Koopman eigenfunctions $\phi_1(x)=s_1 - 2x_2$ and $\phi_2(x)=x_1 + s_2$ with corresponding eigenvalues $\lambda_1=0.8$ and $\lambda_2=-0.5$. The target set is defined in spectral coordinates $\Phi(x) = [\phi_1(x)\ \phi_2(x)]^\top$ as $\mathcal{T}_x=\{x:\|\Phi(x)\|_2\le r\}$, where $J(x)=\|\Phi(x)\|_2^2-r^2$. Simulations are conducted over a time horizon $T=1$ with $r=0.25$, $|u|\le u_{max} = 2.0$, and $|d|\le d_{max} = 0.45$.

For the \textit{approximate rigidity} case, the control and disturbance fields are structured as:
\begin{align*}
    G(x) = \frac{\alpha_u}{d_0}
    \begin{bmatrix} 
        c_2 + 1.4 \\ 
        -1 + 0.7 c_1 
    \end{bmatrix}, 
    E(x) = \frac{\alpha_d}{d_0}
    \begin{bmatrix} 
        0.35 c_2 - 0.5 \\ 
        -0.35 - 0.25 c_1 
    \end{bmatrix},
\end{align*}
where $\alpha_u := 1 + 0.3 s_1 + 0.2 c_2$ and $\alpha_d := 0.8 + 0.2 c_1 - 0.1 s_2$. Following the construction in Section~\ref{sec:approx_rigidity}, the constant spectral input matrices obtained are $B_{u,N} = [1.19\;\; 0.83]^\top$ and $B_{d,N} = [0.34\;\; -0.24]^\top$ by spatial averages over the computational domain. As shown in Figure~\ref{fig:ex1_combined}(d), the resulting Spectral Koopman Hopf (SKH) value function zero-level contour agrees well with the baseline level set method (LSM) toolbox \cite{mitchell2005toolbox} solution; minor discrepancies away from the interface reflect the averaging of state-dependent input fields into constant matrices.

For the \textit{bounds} case, the input fields used are:
\begin{align*}
    G(x) = 
    \begin{bmatrix} 
        c_2 + 1.5 \\ 
        -1 + 0.75 c_1 
    \end{bmatrix}, 
    \quad 
    E(x) = 
    \begin{bmatrix} 
        0.45 c_2 - 0.5 \\ 
        -0.45 - 0.25 c_1 
    \end{bmatrix}.
\end{align*}
To test the Hamiltonian bounds of Section~\ref{sec:reachability_bounds}, we exploit the property $1\le d_0(x)\le 3$ to bound the true Hamiltonian via $\bar{H}^\ell(P,t) \le \bar{H}(X,P,t) \le \bar{H}^u(P,t)$, where,
\begin{equation}
    \begin{aligned}
    \bar{H}^\ell(P,t) &= - 3u_{\max}|\beta_1(P,t)| + d_{\max}|\beta_2(P,t)|, \\
    \bar{H}^u(P,t) &= - u_{\max}|\beta_1(P,t)| + 3d_{\max}|\beta_2(P,t)|,
    \end{aligned}
\end{equation} 
$\beta_1(P,t) = 
\begin{bmatrix}
1 & 0.75
\end{bmatrix}
e^{-\Lambda^\top t}P$
and 
$\beta_2(P,t) = 
\begin{bmatrix}
0.45 & -0.25
\end{bmatrix}
e^{-\Lambda^\top t}P$.

The lower bound $\bar{H}^\ell$ grants maximum control authority and minimum disturbance strength, yielding the value function contours and the over-approximated reachable set $R^\ell$ (white zero-level curve) illustrated in Figure~\ref{fig:ex1_combined}(a). Conversely, the upper bound $\bar{H}^u$ prioritizes the disturbance, producing the under-approximated reachable set $R^u$ shown in Figure~\ref{fig:ex1_combined}(b). Figure~\ref{fig:ex1_combined}(c) directly compares these bounds against LSM baseline (yellow curves), confirming the SKH inner and outer approximations.

\begin{figure}
    \centering
    \includegraphics[width=0.95\linewidth]{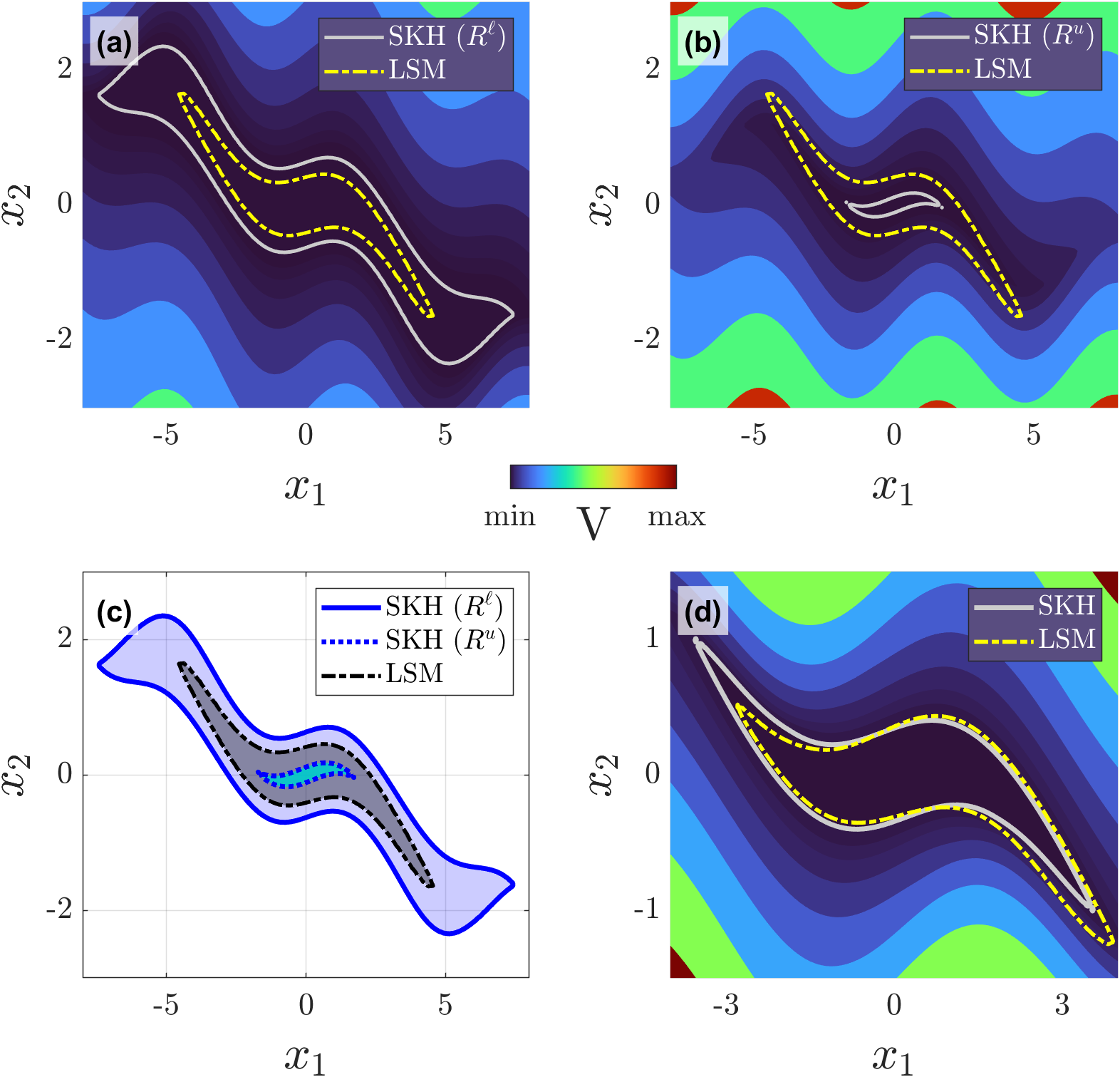}
    \caption{\textcolor{black}{Backward reachable sets for the 2D system: comparison of LSM baseline (yellow) against SKH approximations (white) with respective value function contour plots under (a) $\bar{H}^\ell$, (b) $\bar{H}^u$, and (d) approximate rigidity. (c) Bounds case: reachable set comparison ($R^u \subseteq R \subseteq R^\ell$).}}
    \label{fig:ex1_combined}
\end{figure}

\subsection{N-Link Robotic Arm}
\label{sec:nlink_arm}

The second experiment considers nonlinear robotic-arm dynamics where the Koopman eigenfunctions are not known analytically. The $N$-link system is governed by
\begin{equation}
    M(q)\ddot q+C(q,\dot q)\dot q+K_s(q)+D\dot q=u+d,
\end{equation}
with state $x = \begin{bmatrix} q & \dot q \end{bmatrix}^\top \in \mathbb{R}^{2N}$. Parameters are chosen such that the origin is a locally stable equilibrium of the drift dynamics, and computations are restricted to its neighborhood. 
The Koopman eigenfunctions are approximated via \eqref{eq:phi_decomp}--\eqref{eq:eta_galerkin} as detailed in Section~\ref{sec:eigenfunction_compute}. The terminal data is defined as $J(x)=\|\Phi(x)\|_2^2-r^2$ using the learned eigenfunctions, and the transformed input matrices are approximated by domain averages per \eqref{eq:BuBd_fit}--\eqref{eq:BuBd_sample_approx} from Section~\ref{sec:approx_rigidity}.

For the two-link arm ($N=2, x \in \mathbb{R}^4$), the simulation parameters are $T=0.5, r=0.2$, and $|u_{i}| \le 0.5, |d_{i}| \le 0.2$ for $i \in \{1,2\}$. A grid-based reference solution is computed using the LSM toolbox. Figure~\ref{fig:ex2_combined}(a) illustrates a $q_1$-$q_2$ slice (with $\dot q_1=\dot q_2=0$) comparing an LSM baseline backward reachable set (BRS) (yellow curve) with the Spectral Koopman Hopf (SKH) value function contours and BRS estimate (white zero-level curve). The contours are qualitatively consistent, though the SKH values are generally lower near the interface, producing larger zero sublevel sets. This discrepancy stems from the learned eigenfunction map and averaged spectral input matrices, alongside LSM-inherent factors like finite grid resolution, numerical dissipation, and boundary effects.

For the three-link arm ($N=3, x \in \mathbb{R}^6$), the parameters are $T=1.4, r=0.1$, and $|u_{i}| \le 1, |d_{i}| \le 0.25$ for $i \in \{1,2,3\}$. 
Figure~\ref{fig:ex2_combined}(b) displays the SKH value function contours and BRS estimate on the $q_1$-$q_2$ plane (with $q_3=\dot q_1=\dot q_2=\dot q_3=0$). For this six-dimensional system, the LSM baseline failed due to memory limitations even on a modest grid resolution (21 nodes per dimension). 
The $21^6$ reference grid in this example would contain $85,766,121$ nodes before accounting for storage of auxiliary arrays, whereas SKH does not construct this grid.
These 6D results underscore the primary computational advantage of the SKH method, high-dimensional trajectory data are required only to construct $\Phi_N$ and the averaged spectral input matrices. Once established, the value function is evaluated pointwise, mitigating the curse of dimensionality associated with high-dimensional HJI grids.

\begin{figure}
    \centering
    \includegraphics[width=0.95\linewidth]{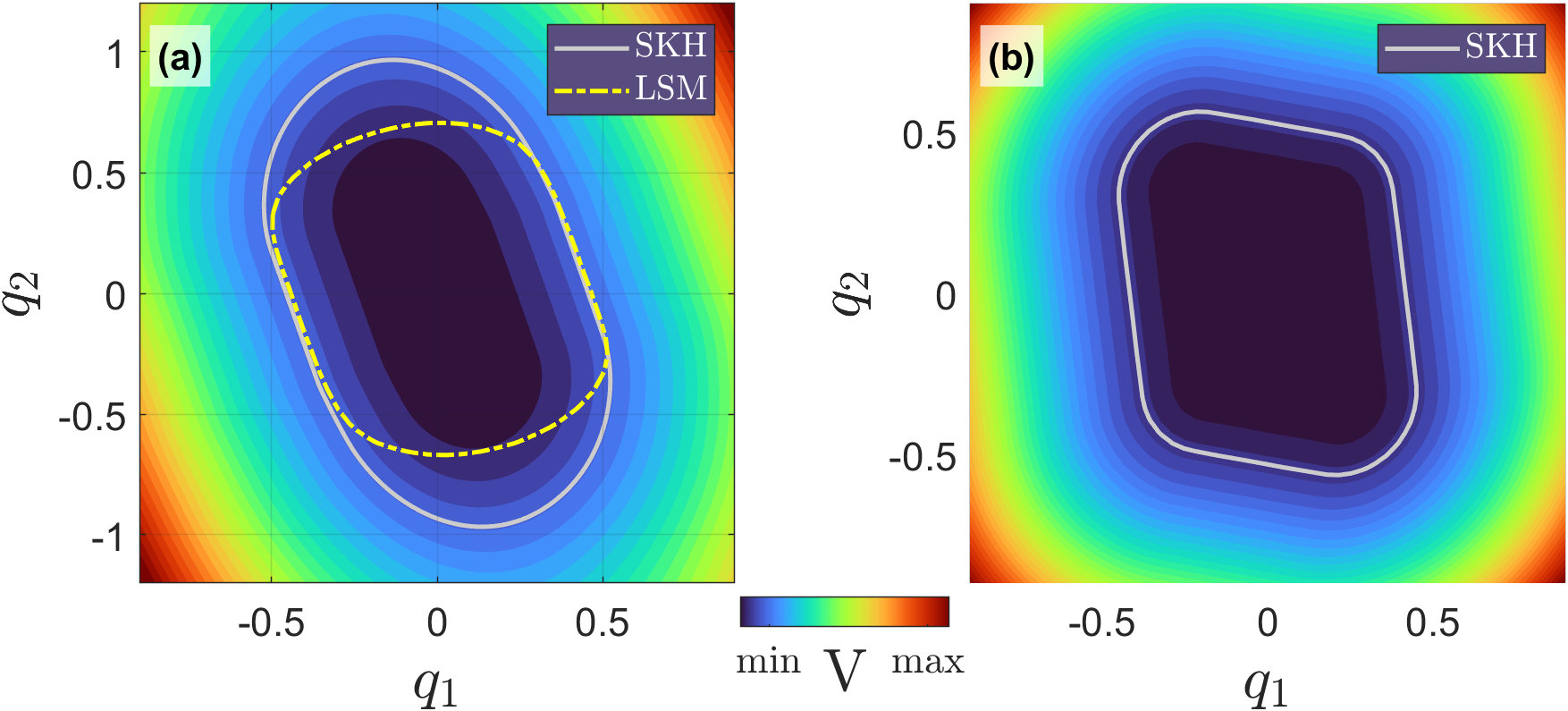}
    \caption{\textcolor{black}{Backward reachable sets for the $N$-link arm: (a) 4D system ($N=2$) SKH value function contours and BRS estimate (white curve) compared against an LSM baseline BRS (yellow curve); (b) 6D system ($N=3$) SKH value function and BRS estimate (white curve) where LSM failed due to memory limits.}}
    \label{fig:ex2_combined}
\end{figure}

\section{Conclusion}

We developed a spectral Koopman-Hopf framework for adversarial reachability analysis of nonlinear systems. By lifting the drift dynamics into Koopman eigenfunction coordinates, the proposed approach replaces state-space discretization of the nonlinear HJI equation by an approximate state-independent, pointwise optimization problem in spectral coordinates. 
The framework employs least-squares approximations of transformed input directions together with Hamiltonian and value-function bounds that yield reachable set enclosures.
The Koopman eigenfunctions are computed using a path-integral formulation, avoiding spatial discretization of the associated eigenfunction PDEs. Numerical examples demonstrate tractable approximation of backward reachable sets for nonlinear adversarial systems while preserving the original system dimension when only the $n$ principal eigenfunctions are retained.

\bibliographystyle{ieeetr}
\bibliography{references}

\end{document}